\documentclass{ifacconf}
\usepackage{amsmath, amssymb, amscd, amsfonts}
\usepackage{graphicx}      
\usepackage{natbib}        
\usepackage{relsize}
\usepackage{algorithm}
\usepackage{algpseudocode}
\usepackage{pifont}
\usepackage{mathrsfs}
\usepackage{xcolor}

\newtheorem{theorem}{Theorem}

\newtheorem{assumption}{Assumption}
\newtheorem{atheorem}{Theorem}

\newenvironment{proof}{\par\noindent{\it Proof. }}{\hfill$\square$\par}
\begin{document}
\begin{frontmatter}

\title{Towards Guaranteed Optimal PID Tuning for Uncertain Nonlinear Systems} 

\thanks[footnoteinfo]{This work was supported by National Key R\&D Program of China Under Grant 2024YFA1013104, and by  National Natural Science Foundation of China under Grant U22B6001, 12288201 and 62303451.}

\author[1]{Jingru Zhu} 
\author[1]{Cheng Zhao} 
\author[1]{Lei Guo}

\address[1]{State Key Laboratory of Mathematical Sciences, AMSS, Chinese Academy of Sciences, Beijing 100190, China.\\
School of Mathematical Sciences, University of Chinese Academy of Sciences, Beijing 100049, China.\\
(e-mail: lguo@amss.ac.cn)}


\begin{abstract}  
Despite the widespread use of PID controllers in engineering practice, designing optimal PID parameters has long been regarded as a challenging problem in both theory and practice, particularly when faced with uncertain nonlinear dynamical systems.
Based on the authors' PID control theory established recently for MIMO nonlinear uncertain systems (Zhao and Guo, 2022), which provides 
a concrete PID parameter set for global stability of PID controlled systems,
this paper further proposes a near-optimal PID tuning method, where only input-output (zeroth-order) data on the control performance is available. 
The tuning method is formulated as a constrained optimization problem and solved by an iterative learning algorithm, referred to as HRS-KW algorithm, that combines a hysteretic random search with the Kiefer–Wolfowitz algorithm, aiming at utilizing the advantages of both global exploration and local gradient acceleration. 
This method operates without requiring precise structural knowledge of the system dynamics, yet 
its almost sure convergence to an $\epsilon$-optimal solution for the PID parameters can be guaranteed in theory while ensuring closed-loop system stability.
Simulation results illustrate that our HRS-KW algorithm outperforms other related optimization methods, exhibiting better convergence 
to the prescribed $\epsilon$-optimal performance set.
\end{abstract}

\end{frontmatter}

\section{Introduction}
\vspace{-8pt}
As is well-known, the proportional-integral-derivative (PID) controller has been widely used in industrial automation, and more than $90\%$ of the control loops employ PID control (\cite{baifenzhijiushi,surveySamad,ZHAO2025Beyond}). 
Despite its prevalence, many PID control loops exhibit unsatisfactory performance due to improper PID control parameters (\cite{2006pidtuning}). 
In practice, poorly tuned PID parameters may lead to oscillations, inferior transient responses, or even closed-loop instability, which can severely degrade system efficiency and safety (\cite{reviewtuning}). 
Therefore, developing systematic and optimal PID tuning methods that can guarantee both closed-loop stability and desired control performance is of great importance, and has attracted sustained and widespread attention from many control theorists and engineering practitioners.  

Over the years, a wide range of PID tuning methodologies have been developed (see e.g., \cite{reviewtuning,dilemmaofpid,stabilityRL,pid-tuning-conventional}).
Conventional tuning techniques include rule-based methods which rely on first-order plus dead-time linear models (e.g., Ziegler–Nichols and Cohen–Coon rules), and optimization-based methods that require an accurate process model (\cite{optimizationbased}).
However, most real-world industrial processes are inherently nonlinear, and subject to uncertainties, which considerably limit the applicability of these model-based approaches (\cite{pid-tuning-conventional}).
In recent years, learning-based and data-driven PID tuning strategies have gained increasing attention (\cite{ESPID,ilc-xjx,RLprocesscontrol}), particularly those built upon  extremum seeking, iterative learning and reinforcement learning.
Although these methods can effectively enhance control performance without explicit model knowledge, rigorous guarantees of closed-loop stability are generally not established, particularly for nonlinear uncertain systems (\cite{pid-tuning-neural-network}).


Recently, we have shown that the classical PID control can ensure global stability for a basic class of nonlinear uncertain systems, provided that the PID parameters are chosen within a three dimensional unbounded stability region (\cite{2017zc,Zhao2018ccc,2022towards}).
However, the problem of designing optimal PID parameters that guarantee both closed-loop system stability and desired control performance remains unresolved, which is the primary motivation of this paper.
To address this problem, we focus on a class of second-order nonlinear uncertain multi-input multi-output (MIMO) systems, 
where the control objective is to design PID control that stabilizes the system while minimizing a performance cost related to output error and control effort.
We formulate the optimal PID tuning as a constrained optimization problem, where the constraint ensures that the PID parameters generated during the tuning process remains in the stability region constructed explicitly in the work (\cite{2017zc,2022towards}). 
It is worth noting that, only the input-output (zeroth-order) data of the control performance can be accessed, which renders traditional gradient or Hessian-based optimization methods inapplicable.

Inspired by \cite{guoWLS1996,wxq,2005threshold-rs}, we propose an iterative learning algorithm in this paper which combines a hysteretic random search with the Kiefer–Wolfowitz algorithm, abbreviated as HRS-KW. 
This algorithm exploits both the global exploration of random search and local gradient acceleration of the KW algorithm.
Without assuming the convexity of the control performance cost, we prove that the proposed HRS-KW algorithm possesses a global convergence property, and converges to an $\epsilon$-optimal solution almost surely for all initial stabilizing PID gains. Such global convergence is mainly attributed to the hysteretic random search, which helps the KW algorithm escape local minima, a phenomenon further illustrated through simulations.

The remainder of this paper is organized as follows. Section II formulates the constrained optimization problem for nonlinear uncertain systems. Section III proposes the HRS-KW algorithm and presents the main results. Section IV presents numerical simulations. Finally, Section V concludes the paper with some remarks.
\vspace{-5pt}
\section{Problem Formulation} 
\vspace{-8pt}
In this paper, we investigate the performance optimization problem of PID control for nonlinear uncertain systems.
For simplicity of presentation, we consider a basic class of nonlinear MIMO system:
\begin{align}\label{system0}
\left\{
\begin{aligned}
    &\dot x_1(t) = x_2(t)\\
    &\dot x_2(t) = f(x_1(t),x_2(t),u(t))\\
    &y(t) = x_1(t)
\end{aligned}
\right.
\end{align}
where $x(t)=\begin{bmatrix}
    x_1^\mathsf{T} (t)&x_2^\mathsf{T}(t)
\end{bmatrix}^\mathsf{T} \in\mathbb{R}^{2n}$ is the state vector, $y(t)\in\mathbb{R}^n$ is the output, $u(t)\in\mathbb{R}^n$ is the input, and $f\in C^1(\mathbb{R}^{3n}, \mathbb{R}^{n})$ is an \emph{uncertain nonlinear} function.

Our control objective is to make the output $y(t)$ converge to a given setpoint $y^*\in\mathbb{R}^n$ using a classical PID control of the form
\begin{align}\label{pid}
u(t)=k_0\textstyle\int_0^te(s)\mathrm{d}s+k_1e(t)+k_2\dot e(t)
\end{align}
where $e(t)=y^*-y(t)$ is the output error and $K:=\begin{bmatrix}
    k_0&k_1&k_2
\end{bmatrix}^\mathsf{T}\in\mathbb{R}^3$ are the tunable control parameters.
At the same time, we aim to minimize the performance cost
\begin{align}\label{cost0}
J(K)=q(e(t_f))+\textstyle\int_0^{t_f} l(e(t),u(t)) \mathrm{d}t
\end{align}
where $l:\mathbb{R}^n\times \mathbb{R}^n \rightarrow \mathbb{R}_{\geq 0}$ denotes the running cost penalizing the output error and control effort, and $q:\mathbb{R}^n \rightarrow \mathbb{R}_{\geq 0}$ is the terminal cost at the given final time $t_f>0$. The functions $l(\cdot)$ and $q(\cdot)$ are required to be  nonnegative and continuous with respect to their respective variables.
In practical engineering applications, the specific forms of $l(\cdot)$ and $q(\cdot)$ are typically determined by the system requirements and design specifications. 
For instance, in the linear quadratic regulator (LQR) problem, a common and standard choice is:
\[
q(e)=e^\mathsf{T} F e,~~l(e,u)=e^\mathsf{T} Q e+u^\mathsf{T} Ru
\] 
with $Q\succeq 0, R\succ 0$, and $F\succeq 0$. Here, for a symmetric matrix $A$, $A\succeq 0$ implies that $A$ is positive semi-definite, and $A\succ 0$ implies that $A$ is positive definite.

We make the following assumption on the uncertain nonlinear function $f$, which is used to quantify the magnitude of system uncertainty and is consistent with the formulation in \cite{2022towards}.
\begin{assumption} \label{assum1}
There exist two positive constants $L_1$, $L_2$ and $\underline{b}$, such that  for all $x_1,x_2,u\in\mathbb{R}^n$,
    \begin{equation}
      \Big\|\frac{\partial f}{\partial{x_1}}\Big\|\le L_1,~\Big\|\frac{\partial f}{\partial{x_2}}\Big\|\le L_2,~
      \frac{1}{2}\Big[\frac{\partial f}{\partial{u}}+\Big(\frac{\partial f}{\partial{u}}\Big)^{\mathsf{T}}~\!\Big]
 \geq \underline{b}I_n
    \end{equation}
    where $\frac{\partial f}{\partial{x_i}}$ and $\frac{\partial f}{\partial{u}}$ are the $n\times n$ Jacobian of $f$ with respect to $x_i$ and $u$, $I_n$ is the $n\times n$ identity matrix.
\end{assumption}

Under Assumption 1, it has been demonstrated that the PID control can ensure global stability of the closed-loop system, with the three PID parameters freely chosen from
a 3-dimensional, unbounded stability region. To be specific, we have the following.
\begin{atheorem}[{\cite{2022towards}}]
Consider the PID controlled nonlinear uncertain system (1)-(2), where function $f$ satisfies Assumption 1. Suppose the PID parameters are selected from the following 3-dimensional open and unbounded set:
\begin{align}\label{a1-pid}
\Omega_{{\rm pid}}=\left\{K\in\mathbb{R}_+^3 \left|
		~k_p^2>2k_ik_d+\bar k,
		~k_d^2~\!>k_p/\underline b~\!+\bar k
\right.
\right\},
\end{align}
where $\bar k :=(L_1+L_2)(k_p+k_d)/\underline b$.
Then the solution of the closed-loop system will satisfy 
\begin{align*}
\lim_{t\to\infty} \|e(t)\|+\|\dot e(t)\|=0\end{align*}
with an exponentially fast rate,
for any setpoint $y^*\in\mathbb{R}^n$ and any initial states $x(0)\in \mathbb{R}^{2n}$.
\end{atheorem}

Building upon Theorem A1, it is natural to consider the following constrained optimization problem with guaranteed closed-loop stability:
\begin{subequations} \label{op} 
\begin{align}
     \min_{K\in \Omega_{{\rm pid}}} ~ J(K)&=q(e(t_f))+\textstyle\int_0^{t_f} l(e(t),u(t))  \mathrm{d}t, \label{opa}
     \\
    \text{s.t.}~\dot x_1(t) &= x_2(t) \label{opb}\\\ 
    \dot x_2(t) &= f(x_1(t),x_2(t),u(t))\label{opc}\\
    x(0)&=x_0 \label{opd}\\
    u(t)&=k_0\textstyle\int_0^te(s)\mathrm{d}s+k_1e(t)+k_2\dot e(t) \label{ope}\\
    e(t)&=y^*-x_1(t) \label{opf}
\end{align}
\end{subequations}
In the constrained optimization problem above, we have restricted the PID parameters to be selected from the stability region 
$\Omega_{{\rm pid}}$. This is a quite reasonable assumption, as stability is, of course, one of the most fundamental requirements for a control system. By restricting the PID parameters to this region, one can effectively guarantee the stability of the closed-loop system, ensuring that the system can operate reliably under various conditions.

\emph{Information availability:} In our setting, the optimization solver is allowed to query the value of $J(K)$ for a given PID control parameter $K$, that is, it has access solely to input-output (zeroth-order) information of $J(\cdot)$.
However, the first-order (gradient) or second-order (Hessian) derivatives cannot be queried, since the nonlinear function $f(\cdot)$ in the dynamical system \eqref{opb} is uncertain, and  $q(\cdot),~l(\cdot)$ may also be unknown.

\textbf{Remark 1.}
We emphasize that the problem \eqref{op} is challenging due to two major obstacles. First, the objective function $J$ is \emph{unknown}, which makes classical optimization methods that rely on gradient or Hessian information inapplicable, and stochastic approximation like algorithms need to be used.
Second, it should be noted that the objective function $J$ may be highly nonlinear and non-convex, as we have only assumed that $l$ and $q$ in \eqref{opa}
are nonnegative continuous functions. 
In fact, even when dynamic system \eqref{opb}-\eqref{ope} is linear and, both $l$ and $q$ are quadratic forms, the convexity of $J$ may still not be guaranteed (see Example 1). Consequently, traditional gradient-based optimization methods risk getting trapped in local minima or saddle points, and so certain global optimization techniques need to be considered.
\vskip 0.2cm
\textbf{Example 1. }[\emph{Nonconvexity of the Performance Cost}] Consider the optimization problem \eqref{op}, where $x_1,x_2$, $u$ in \eqref{opc} are all scalars, the function $f$ is given by $f(x_1,x_2,u)=ax_1+bx_2+u$, and the performance cost is $$J(K)=\textstyle\int_0^{t_f} \left(\lambda|e(t)|^2+|u(t)|^2\right) \mathrm{d}t.$$ 
For such a linear system under PID control, it is straightforward to deduce (using the Routh-Hurwitz criterion) that the closed-loop system is globally stable if and only if $k_1>a,k_0>0,(k_1-a)(k_2-b)>k_0$.

We now turn our attention to the cost function $J(K)$, which can be expressed (in this example) as follows:
\begin{align}\label{example-J}
\begin{split}
    J= z^{\mathsf{T}}_0Qz_0-z^{\mathsf{T}}_{1}Qz_{1}+z_*^{\mathsf{T}} P\left[ t_f z_*+2A_c^{-1}(z_{1}-z_0)\right]
\end{split}
\end{align}
where $P=KK^\mathsf{T}+\lambda \text{ diag}\{0,1,0\}$, $ z_{1}=e^{t_fA_c}z_0$ with
\begin{align}\label{Ac}
A_c=\begin{bmatrix}
        0&1&0\\
        0&0&1\\
        0&a&b
    \end{bmatrix}-\begin{bmatrix}0\\0\\1\end{bmatrix}K^{\mathsf{T}},~~z_0=\begin{bmatrix}
    \textstyle0\\e(0)\\\dot e (0)
\end{bmatrix}-\underbrace{A_c^{-1}\begin{bmatrix}
    0\\0\\ay^*
\end{bmatrix}}_{\text{denoted as } z_*}
\end{align}
and $Q$ solves the Lyapunov equation $A_c^\mathsf{T} Q+QA_c+P=0$. 
It can be observed from (\ref{example-J}) that the cost function $J$ is highly nonlinear, and its gradient with respect to 
$K$ is quite complicated, even in the absence of system uncertainties. 
Besides, it can be verified that $J$ may not be a convex function of $K$. The detailed derivations of equation (\ref{example-J}) and the analysis of the nonconvexity of $J$ are provided in the Appendix A.
\vspace{-5pt}
\section{Near-Optimal PID Tuning}
\vspace{-8pt}
Note that the PID parameter set \eqref{a1-pid} is unbounded (its Lebesgue measure is infinite), which poses a challenge for random search methods, as they typically require a compact domain in order to define a uniform sampling distribution.
To address this issue, we introduce the following feasible set:
\begin{align}\label{feasible set}
\begin{split}
S = \big\{K \in \mathbb{R}_{+}^3 \left|\right. \|K\|\le R,~&k_p^2\geq 2k_ik_d+\bar k+r,\\
		~&k_d^2~\!\geq k_p/\underline b~\!+\bar k+r\big\}    
\end{split}
\end{align}
where $r$ and $R$ are two positive constants. 
It is clear that $S$ is a compact set and $S \subseteq \Omega_{{\rm pid}}$. 
Besides, for any $K\in \Omega_{{\rm pid}}$, it holds that $K\in S$ as long as $r$ is suitably small and $R$ is suitably large. Therefore, the feasible set \eqref{feasible set} is a compact inner approximation of $\Omega_{{\rm pid}}$.

Since the exact gradient of the objective function \eqref{cost0} is unavailable, we estimate it using the KW algorithm, which approximates the gradient through finite differences.
Specifically, the gradient estimate of $J$ at $K$ in the KW algorithm is defined as follows:
\begin{align}\label{tiduguji}
    &\widehat{\frac{\partial J}{\partial K}}=
    \frac{1}{2c}\begin{bmatrix}
    J(K+c e_0)-J(K-c e_0)\\
    J(K+c e_1)-J(K-c e_1)\\
    J(K+c e_2)-J(K-c e_2)
    \end{bmatrix}
\end{align}
where $c>0$ is a perturbation constant and
\begin{align*}
e_0=\begin{bmatrix}
        1&0&0
    \end{bmatrix}^\mathsf{T},~e_1=\begin{bmatrix}
        0&1&0
    \end{bmatrix}^\mathsf{T},~e_2=\begin{bmatrix}
        0&0&1
    \end{bmatrix}^\mathsf{T}.
\end{align*}    

\textbf{HRS-KW Algorithm: } Let us introduce an iterative learning algorithm, referred to HRS-KW, which integrates hysteretic random search with the KW method and is recursively defined by 
\begin{equation}\label{RSKW}
    K_{i}=
    \begin{cases}
        \eta_{i},&\text{if }J(\eta_{i})\leq J(K_{i-1})-\tfrac{\epsilon }{2},~~J(\eta_{i})\leq J(w_{i})\\
        w_{i},&\text{if }J(w_{i})\leq J(K_{i-1})-\tfrac{\epsilon }{2},~J(w_{i})<J(\eta_{i})\\
        K_{i-1},&\text{otherwise}
    \end{cases}
\end{equation}
where $\epsilon>0$ is a given descent threshold, the initial $K_0$ is chosen from $\Omega_{\mathrm{pid}}$, $\{\eta_i\}$ is an independent and identically distributed random sequences, uniformly distributed over $S$, and 
\begin{align} \label{touying}
w_{i}=\mathlarger{\Pi}_S\big\{K_{i-1} - \alpha G_{i-1}\big\}
\end{align}
where $\alpha$ is the learning rate, $G_{i-1}=\widehat{\frac{\partial J}{\partial K}}\big|_{K_{i-1}}$ is the  estimate of the gradient of $J$ at $K_{i-1}$,  and $\Pi_S\{\cdot \}$ is the projection operator onto $S$.

\vspace{1em}
\begin{theorem}\label{th1}
Consider the constrained optimization problem \eqref{op}, where the nonlinear uncertain function $f$ satisfies Assumption \ref{assum1}.
If the HRS-KW algorithm \eqref{RSKW} is applied  with any initial value $K_0\in \Omega_{\mathrm{pid}}$,
then for any initial state $x_0$ and any setpoint $y^*$,
we have $
\lim_{t\rightarrow \infty}K_t=K_{\infty},~ a.s.$,
where $K_{\infty}$ belongs to the set 
\[
S(\epsilon):=\{K\in S\left| J(K)\leq J^*+\epsilon \right.\},
\]
and $J^*:=\min_{K\in S}J(K)$ denotes the optimal value over the feasible set $S$.
\end{theorem}
\begin{proof} 
	We begin by demonstrating that the objective function $J(K)$ is well-defined and continuous on the set $S$. 
To this end, we introduce the auxiliary state $x_0(t)=\int_0^t(y^*-x_1(s))\mathrm{d}s$. Then the PID control \eqref{pid} can be expressed (in terms of the augmented state vector $\bar x:=\begin{bmatrix}
    x_0&x_1&x_2
\end{bmatrix}^\mathsf{T}$) as $$u(t)=k_0x_0(t)-k_1x_1(t)-k_2x_2(t)+k_1y^*$$ and the closed-loop system \eqref{system0}-\eqref{pid} can be rewritten as the autonomous differential equation 
\begin{align}\label{zizhi}
    \dot{\bar{x}}  = F(\bar x, K), ~\bar x(0)=\begin{bmatrix}
        0&x^\mathsf{T}(0)
    \end{bmatrix}^\mathsf{T}
\end{align}
where the vector field $F(\bar x, K)$ is given by
\begin{align*}
    F(\bar x,~K)=\begin{bmatrix}
        y^* - x_1\\x_2\\f(x_1,x_2,k_0x_0-k_1x_1-k_2x_2+k_1y^*)
    \end{bmatrix}.
\end{align*} 
It is easy to see that $F(\bar x,K)$ is continuously differentiable in $(\bar x,K)$, due to the fact $f\in C^1$. Denote the solution of \eqref{zizhi} as $\bar x(t,K)$, then it is a continuous function of $K$.
Therefore, both $e(t)=y^*-x_1(t)$ and $u(t)=k_0x_0(t)-k_1x_1(t)-k_2x_2(t)+k_1y^*$ depend continuously on $K$.
Besides, since $K\in S\subset \Omega_{{\rm pid}}$, Theorem A1 tells us $e(t)$ converges to zero and $u(t)$ is a bounded function on $[0,~\infty)$. Hence $J(K)=q(e(t_f))+\int_0^{t_f} l(e(t),u(t)) \mathrm{d}t$ is finite for any $t_f>0$.
Note also that $l$ and $q$ are continuous,  we conclude that $J(K)$ is a continuous function of $K$ over $S$. 

We next prove that both $\{K_i\}$ and $\{J(K_i)\}$  converges.

By the update law of HRS-KW algorithm \eqref{RSKW}, we know that $\{J(K_i)\}$ is a decreasing and nonnegative sequence. Consequently, $\{J(K_i)\}$  converges to some  nonnegative limit  $J_{\infty}$.
To show that $\{K_i\}$ converges, we first denote $$I_0=\{i\geq0: K_{i+1}\neq K_i\}.$$
According to \eqref{RSKW}, for any $i\in I_0$, $J(K_{i+1})\leq J(K_i)-\epsilon /2$. 
Combine this with the fact that $J(K_i)$ is decreasing implies that $|I_0|\leq 2(J(K_i)-J^*)/ \epsilon$, where $|I_0|$ denotes the cardinality (i.e., the number of elements) of the set $I_0$. 

Let $i_0$ be the largest element in $I_0$. By the definition of $I_0$, we know that $K_i \equiv K_{i_0}$ for all $i\geq i_0$.
Therefore, $\{K_i\}$ converges to $K_{i_0}$ in a finite number of steps, with the limit $K_\infty=K_{i_0}\in S$.

We end the proof of Theorem \ref{th1} by showing that, with the assistance of random search, the limit $K_\infty$ (recall it is $K_{i_0}$) satisfies the inequality $J(K_\infty)\leq J^* +\epsilon$. 
 
Note that the feasible set $S$ is compact, there exists some $K^*\in S$, such that $J(K^*)=\min_{K\in S}J(K)$.
Since $J$ is continuous in $K$, there exists $\delta>0$ such that for all $K\in B(K^*,~\delta) \cap S$, we have $0\le J(K)-J^*\leq \epsilon /2$. Furthermore, note that the boundary of the $S$ is composed of smooth surfaces, so the Lebesgue measure of $B(K^*,~\delta) \cap S$ is positive.

Define a sequence of random event $D_i=\{\eta_i \in B(K^*,~\delta) \cap S\}$, $i\geq 1$. Note the $\{\eta_i\}$ is an i.i.d. random sample, then we have 
\begin{align*}
P\left(\bigcap_{i=1}^{N}D_i^c\right)&=\prod_{i=1}^NP(D_i^c) =(P(D_1^c))^N\\
    &=\left[1-\frac{\mu(B(K^*,~\delta) \cap S)}{\mu(S)}\right]^N\rightarrow 0 ~(N \rightarrow \infty)
\end{align*}
where $\mu(\cdot)$ denotes the 3-dimensional Lebesgue measure.
Hence, $P\left(\bigcap_{i=1}^{\infty}D_i^c\right)=0$, which yields  
\begin{align}\label{p1}
P\left(\bigcup_{i=1}^{\infty} D_i\right)=P\big\{\exists i\geq 1: \eta_i \in B(K^*,~\delta)\cap S\big\}=1.\end{align}
Define the first hitting time $$\tau=\min\{t\geq 1:~\eta_{t}\in B(K^*,~\delta) \cap S \}.$$ Then it follows from (\ref{p1}) that $\tau < \infty$ a.s.. 

In the following, we show that $J(K_\tau)\le J^*+\epsilon$ by considering two cases. 
By the definitions of $\tau$ and $B(K^*,\delta)$, it is easy to see that $J(\eta_\tau)\le J^*+\epsilon/2.$ 

Case 1:  $\min\{J(\eta_{\tau}),J(w_\tau)\}\leq J(K_{\tau-1})-\epsilon /2$, then by the update law \eqref{RSKW}, we know that  $J(K_{\tau})\leq J(\eta_{\tau}) \leq J^*+\epsilon /2$.

Case 2: Both $J(\eta_{\tau})> J(K_{\tau-1})-\epsilon /2$ and $J(w_{\tau})> J(K_{\tau-1})-\epsilon /2$. By the update law \eqref{RSKW},  then  $J(K_{\tau})=J(K_{\tau-1})<J(\eta_{\tau} )+ \epsilon /2  \le  J^* + \epsilon$. 

Both two cases imply that $J(K_\tau)\le J^*+\epsilon$. Recall that the sequence $\{J(K_i)\}$ is decreasing, we deduce that
$$J(K_\infty)= \lim_{i\to\infty}J(K_i)\le J(K_\tau)\le J^*+\epsilon.$$ Thus,  $K_\infty$ belongs to the set $\{K\in S\left| J(K)\leq J^*+\epsilon \right.\}$.
\end{proof}
\textbf{Remark 2.} Theorem \ref{th1} shows that the HRS-KW algorithm converges almost surely to the $\epsilon$-optimal set for any initial $K_0\in \Omega_{\mathrm{pid}}$, even in the presence of large-scale uncertainty of function $f$ in \eqref{opc}. From the proof, it can be seen that the global convergence (i.e.,  $K_0\in \Omega_{\mathrm{pid}}$ can be chosen arbitrarily) is largely attributed to the exploration capability of its hysteretic random search steps, which, unlike deterministic gradient-based methods, allow the algorithm to avoid getting stuck in local minima and saddle points. 
Finally, we point out that the descent threshold $\epsilon$ in the algorithm affects both the convergence rate of sequence $\{K_i\}$ and the ultimate cost value $J(K_{\infty})$, thereby making its choice a critical factor for the optimization efficiency and solution quality, and a topic worthy of further investigation (\cite{wxq}).

We emphasize that Theorem A1, regarding the global stability of the PID control system, is only a sufficient condition. A more critical and challenging theoretical  problem lies in determining a larger and necessary stability region for the PID parameters. Solving this issue would have meaningful theoretical and practical implications for two reasons: 1) it ensures that the optimal PID parameters are not ``overlooked" during the search for the best controller settings; 2) smaller control gains often offer advantages in addressing issues such as controller saturation and the amplification of measurement noise.

It is known that, for a specific class of SISO nonlinear systems, a sufficient and necessary condition for the PID parameter selection is provided in \cite{2017zc} (See Proposition 1). Specifically, for the following SISO nonlinear uncertain system:
\begin{align}\label{system-2017}
\left\{
\begin{aligned}
    &\dot x_1(t) = x_2(t)\\
    &\dot x_2(t) = f(x_1(t),x_2(t))+u(t)\\
    &y(t) = x_1(t).
\end{aligned}
\right.
\end{align} 
where $x_1,x_2$ and $u$ are all scalars, and $f$ is a nonlinear uncertain function.
Suppose function $f$ belongs to the following function space $\mathcal{G}_{L_1,L_2}$, which is defined by  
\begin{align*}
\Big \{f\in C^{2}~\!\Big|~\!\frac{\partial{f}}{\partial{x_1}}\le L_1,\frac{\partial{f}}{\partial{x_2}}\le L_2,\frac{\partial^2f}{\partial x_2^2}=0,~x\in\mathbb{R}^2\Big\},
\end{align*}
where $L_1>0,~L_2>0$ are constants and   $C^{2}(\mathbb{R}^2)$ is the space of  twice continuously  differentiable functions.
\begin{atheorem}
For any setpoint $y^*\in\mathbb{R}$ and any $f\in \mathcal{G}_{L_1,L_2}$, 
the closed-loop system \eqref{system-2017}-\eqref{pid} satisfies $\lim_{t\rightarrow \infty} x_1(t)=y^*,~\lim_{t\rightarrow \infty} x_2(t)=0$ for all initial state $x(0)\in\mathbb{R}^2$
if and only if the PID parameter $K$ belongs to the set
\begin{align}
    \Omega=\left\{K\left |~k_0>0,~\bar k_1>0,~\bar k_1\bar k_2>k_0\right.\right\},
\end{align}
where $\bar k_1=k_1-L_1$, $\bar k_2=k_2-L_2$.
\end{atheorem}

Based on Theorem A2, and 
 defining the feasible set 
$$\Omega_{r,R}=\left\{K\left | ~\|K\|\le R,~ k_0\geq r,~\bar k_1\geq r,~\bar k_1\bar k_2-k_0\geq r\right.\right\},$$
we now present the following result, the proof of which is analogous to that of Theorem \ref{th1}.
\vskip 0.1cm
\begin{theorem}\label{th2}
Consider the dynamic optimization problem
\eqref{op}, where the nonlinear uncertain function $f\in\mathcal{G}_{L_1,L_2}$ and the feasible set is $\Omega_{r,R}$.
If the HRS-KW algorithm \eqref{RSKW} is applied with any $K_0\in \Omega$,
then for any initial state and any setpoint,
we have $$\lim_{t\rightarrow \infty}K_t=K_{\infty}\in \{K\in \Omega_{r,R}\left| J(K)\leq J^*+\epsilon \right.\},$$
where $J^*:=\min_{K\in \Omega_{r,R}}J(K)$ is the optimal value.
\end{theorem}

\section{Simulations}
\vspace{-10pt}
In this section, we present numerical simulations of the proposed HRS-KW algorithm, and compare its performance against both the RS and KW algorithms.

In the first simulation, we let the function $f(x_1,x_2,u)=2\sin x_1+x_2+u$ in the optimization problem \eqref{op} with performance cost $\int_0^{T}e^2(t)+u^2(t) \mathrm{d}t$. 
According to Theorem \ref{th2}, we choose the following feasible set:
\begin{align*}
S_1 &= \left\{K\in\mathbf{R}^3 \big| \|K\| \leq 20,~ k_0\geq 0.001,~k_1 \geq 2.001, \right. \\
&~~~~~~~~~~~~~~~~~ \left. (k_1-2)(k_2-1)\geq k_0+0.001 \right\}.
\end{align*}
\vspace{-15pt}
\begin{figure}[H]
    \centering
    \includegraphics[width=1\linewidth]{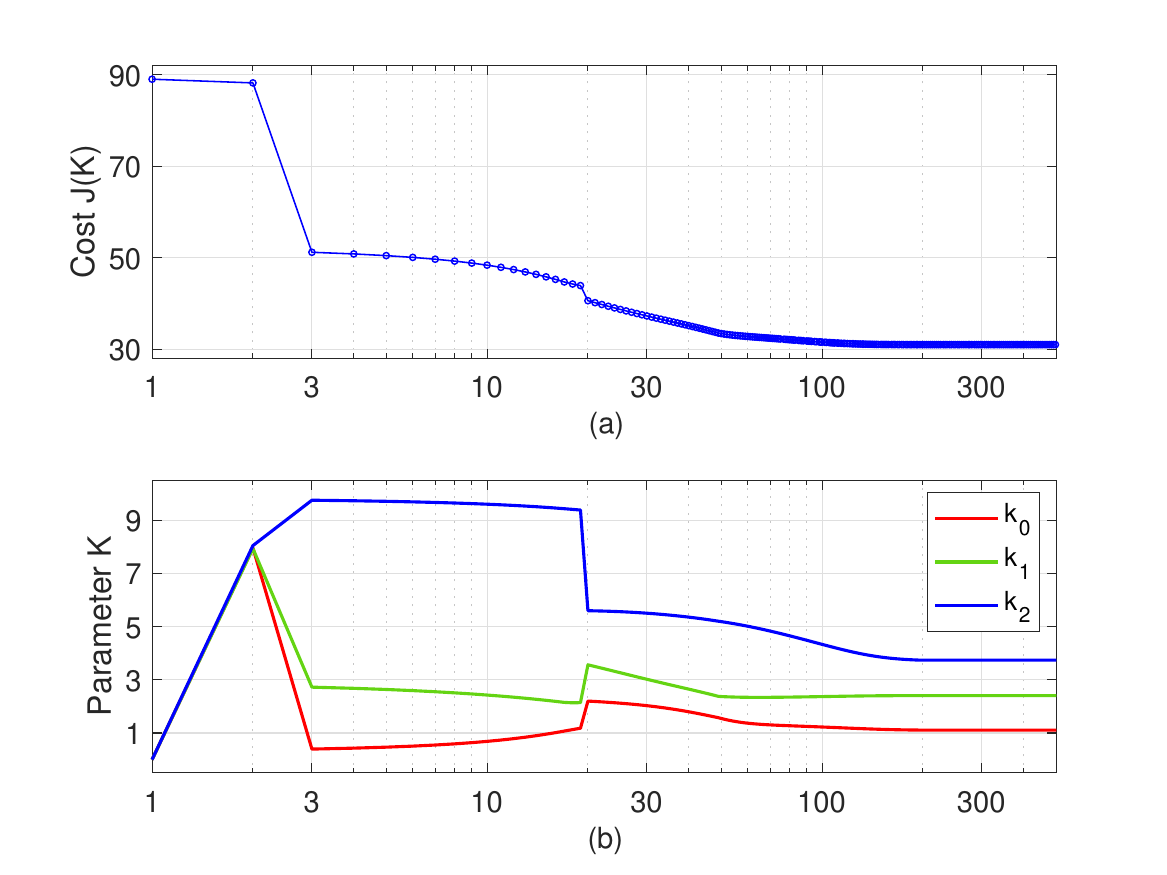}
    \caption{Evolution of $J$ and $K$ under HRS-KW algorithm. The time horizon is set to $T = 5$, $y^* = 1$, $x(0) = [3~2]^\mathsf{T}$. In the KW step, the perturbation constant $c = 0.001$, the learning rate $\alpha = 0.01$, the descent threshold $\epsilon = 0.0001$, the initial PID gain $K_0 = [8~8~8]^\mathsf{T}$.}
    \label{fig:rskw-simu-1}
\end{figure}
\vspace{-18pt}
\begin{figure}[H]
    \centering
    \includegraphics[width=1\linewidth]{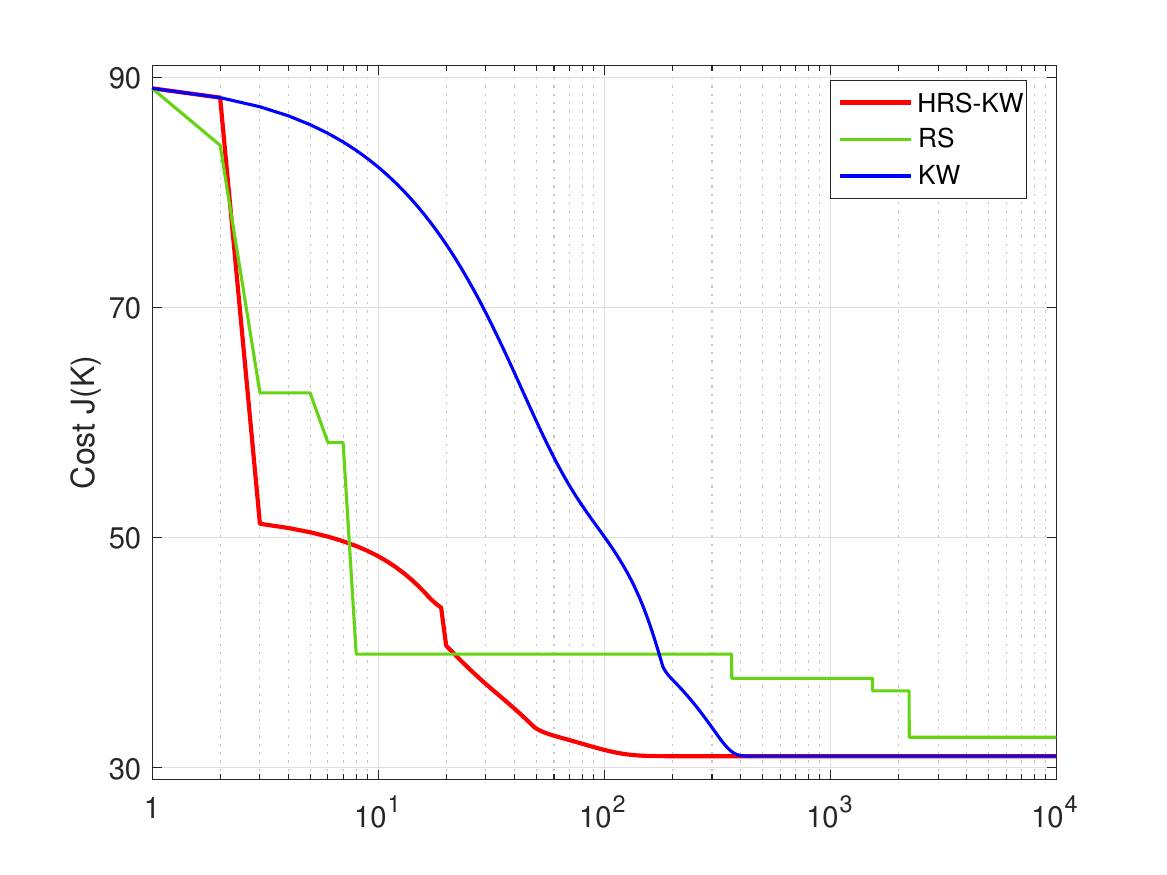}
    \caption{Evolution of $J$ under HRS-KW, RS, KW algorithms. The simulation parameters are the same as Fig. \ref{fig:rskw-simu-1}.}
    \label{fig:duibi-simu-1}
\end{figure}
\vspace{-10pt}
Fig. \ref{fig:rskw-simu-1} illustrates the evolution of cost $J(K)$ and parameter $K$ under HRS-KW algorithm. 
We can see that both cost $J$ and $K$ converge within approximately $250$ iterations, with limits $K^*=[1.1097~2.4056~3.7384]^\mathsf{T}$ and $J(K^*)=30.9969$. Notably, the cost curve exhibits several steep descents, which are attributed to the random search step. Fig. \ref{fig:duibi-simu-1} illustrates that the HRS-KW algorithm converges faster than both the RS and KW algorithms. Both HRS-KW and KW eventually converge to the same value; however, aided by random search, our proposed HRS-KW algorithm exhibits several steep descents, which significantly accelerate the convergence process.
In contrast, the RS algorithm exhibits a much slower convergence process, with many updates failing to yield effective descent.
In the second simulation, we consider the optimization problem \eqref{op} with $f(x_1,x_2,u)=x_1 - x_1^3 - 0.2x_2+u$ and the performance cost $50e^2(T)+\int_0^T (e^2(t)+ 2u^2(t))  \mathrm{d}t$. We choose the feasible set (see Theorem \ref{th2})
\begin{align*}
S_2 &= \left\{K\in\mathbf{R}^3 \big| \|K\| \leq 20,~ k_0\geq 0.001,~k_1 \geq 2.001, \right. \\
&~~~~~~~~~~~~~~~~~ \left. (k_1-1)(k_2+0.2)\geq k_0+0.001 \right\}.
\end{align*}

From Fig. \ref{fig:duibi-simu-2}, we find that the KW algorithm can indeed become trapped in a local minimum of $J=12.2430$ (the corresponding PID gain $K=[0.0348~1.0670~0.3345]^\mathsf{T}$). 
\begin{figure}[H]
    \centering
    \includegraphics[width=1\linewidth]{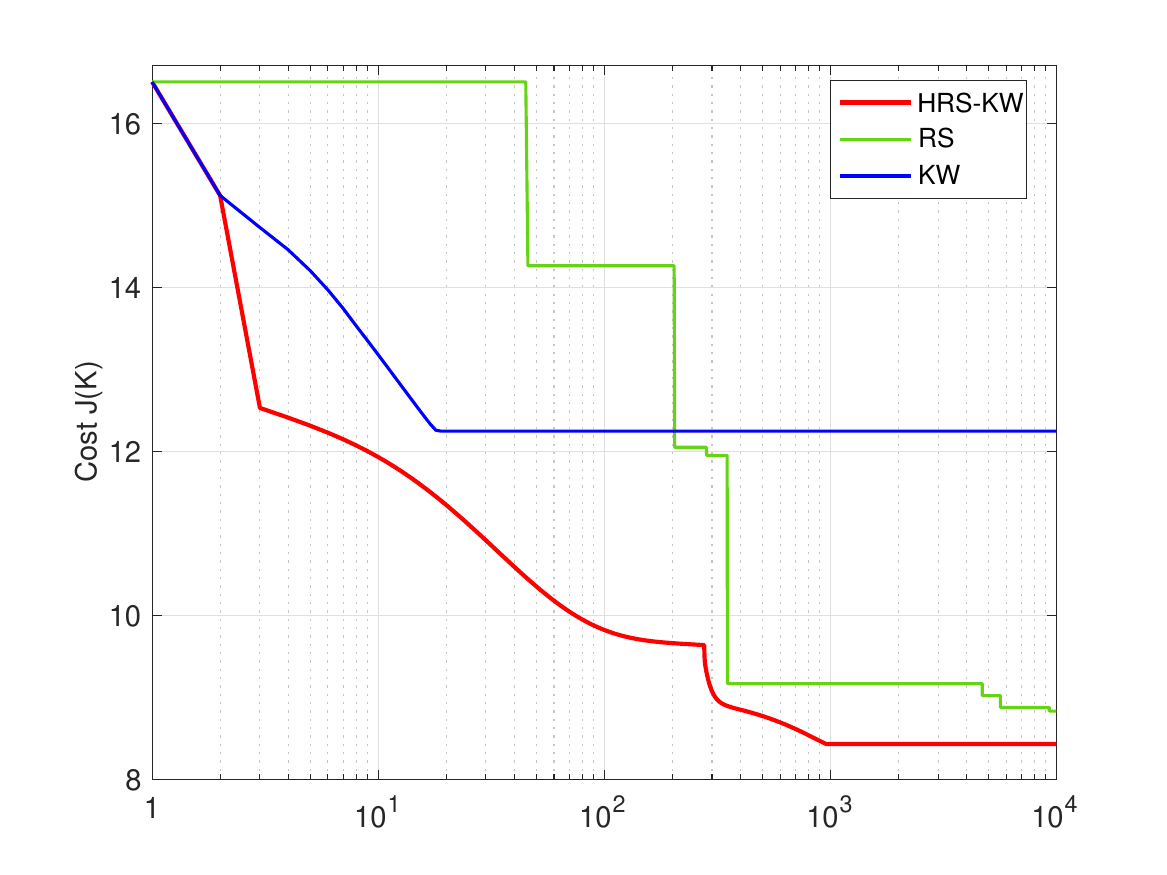}
    \caption{Evolution of $J$ under HRS-KW, RS and KW algorithms. The time horizon is set to $T = 3$, $y^* = 1$, $x(0) = [-1.2~0]^\mathsf{T}$. In the KW step, the perturbation constant $c = 0.001$, the learning rate $\alpha = 0.01$, the descent threshold $\epsilon = 0.0001$, the initial PID gain $K_0 = [0.4~1.5~0.8]^\mathsf{T}$.}
    \label{fig:duibi-simu-2}
\end{figure}
\vspace{-15pt}
\section{Conclusion}
This paper studies the optimal PID tuning for a class of nonlinear uncertain systems.  
Our control objective is to drive the system output to a prescribed setpoint while minimizing a general control performance cost that penalizes output errors and control effort over a finite horizon. 
We propose an iterative learning algorithm, termed HRS-KW algorithm, which does not require any structural information about the uncertain dynamical systems.  
We show that the proposed algorithm converges almost surely to an $\epsilon$-optimal solution for the PID parameters without assuming the convexity of the performance cost, while guaranteeing closed-loop stability throughout the learning process.
For future investigation, it would be interesting to discuss the convergence speed of our algorithm, optimize the design of its key parameters (e.g., the descent threshold $\epsilon$, the learning rate $\alpha$, etc.), and the integration with other optimization methods to further accelerate the tuning process.

\bibliography{ifacconf}             
\vspace{10pt}
\appendix
\vspace{10pt}
\section{Analysis of Example 1}
\vspace{17pt}
First, we derive the equality \eqref{example-J}.
Let us denote
\begin{align*}
z_*=A_c^{-1}\begin{bmatrix}
    0\\0\\ay^*
\end{bmatrix},~~
z(t)=\begin{bmatrix}
    \textstyle\int_0^te(s)\mathrm{d}s\\e(t)\\\dot e(t)
\end{bmatrix}-z_*,    
\end{align*} 
then $u(t)=K^\mathsf{T}(z(t)+z_*)$ and $\dot z(t)= A_c z(t)$, where $A_c$ is given by \eqref{Ac}.
Recall $P=KK^\mathsf{T}+\lambda \text{ diag}\{0,1,0\}$, we have
\begin{align*}
J(K)=&\textstyle\int_0^{t_f} \lambda|e(t)|^2+|u(t)|^2 \mathrm{d}t\\=&\textstyle\int_0^{t_f}(z(t)+z_*)^\mathsf{T} P(z(t)+z_*)\mathrm{d}t\\
=&\textstyle\int_0^{t_f} z^\mathsf{T}(t)P z(t) \mathrm{d}t+2\int_0^{t_f} z^\mathsf{T}_* P z(t)\mathrm{d}t+t_fz^\mathsf{T}_*Pz_*.
\end{align*}
Since $A_c^\mathsf{T} Q+QA_c+P=0$, we have 
\begin{align*}
    &\textstyle\int_0^{t_f} z^\mathsf{T}(t)Pz(t) \mathrm{d}t\\=&-\textstyle\int_0^{t_f} z^\mathsf{T}(t)(A_c^\mathsf{T} Q+QA_c)z(t) \mathrm{d}t\\
    =&-z^\mathsf{T}(t)Qz(t)\big|_{0}^{t_f}=z^\mathsf{T}(0) Qz(0)-z^\mathsf{T}(t_f)Qz(t_f)\\
    =&z_0^\mathsf{T} Qz_0-z_1^\mathsf{T} Qz_1,
\end{align*}
where we have used the facts  that $z_0=z(0)$ and $z_1=e^{t_fA_c}z_0=z(t_f)$ in the last equality.
Besides, 
\begin{align*}
\textstyle\int_0^{t_f} z^\mathsf{T}_* P z(t)\mathrm{d}t=&z^\mathsf{T}_* P\textstyle\int_0^{t_f}e^{A_c t}\mathrm{d}tz(0)\\=&z_*^\mathsf{T} P A_c^{-1}(e^{t_f A_c}-I)z(0)\\
=&z_*^\mathsf{T} P A_c^{-1}(z_1-z_0),
\end{align*}
 we finally obtain \eqref{example-J}.

 We next show that $J$ is not convex when
\[
a=b=-1,~\lambda=1,~y^*=1,~t_f=5,~x_1(0)=1,~x_2(0)=0.
\]
For this case,  the PID parameter set is 
\[
\Omega_{0}=\{K|k_0>0,~k_1>-1,~(k_1+1)(k_2+1)>k_0\}.
\]
Let us consider two sets of parameters
\begin{align*}
    &K_1=\begin{bmatrix}
        0.2&0.9&0.5
    \end{bmatrix}^\mathsf{T},
    ~K_2=\begin{bmatrix}
        6.0&0.4&3.6
    \end{bmatrix}^\mathsf{T},\\
    &K_3=\begin{bmatrix}
        0.5&2.4&1.1
    \end{bmatrix}^\mathsf{T},
    ~K_4=\begin{bmatrix}
        0.8&3.1&-0.9
    \end{bmatrix}^\mathsf{T}
\end{align*}
and their midpoints are 
\begin{align*}
 &\bar K=(K_1+K_2)/2=
\begin{bmatrix}
    3.10&0.65&2.05
\end{bmatrix}^\mathsf{T},\\
 &\bar{\bar K}=(K_3+K_4)/2=\begin{bmatrix}
     0.65&2.75&0.1
 \end{bmatrix}^\mathsf{T}.   
\end{align*}
These six sets of parameters are all in $\Omega_{0}$. 
Through numerical calculations, we can obtain that
\begin{align*}
&J(K_1)\approx3.6597,~J(K_2)\approx10.4196,\\
&J(K_3)\approx5.5962,~J(K_4)\approx34.1031,\\
&J(\bar K)\approx 8.7438 ,~J(\bar{\bar K})\approx8.2018,\\
    &J(\bar K)>(J(K_1)+J(K_2))/2\approx 7.0396,\\
     &J(\bar{\bar K})< (J(K_3)+J(K_4))/2\approx19.8496.    
\end{align*}
Hence, $J(K)$ exhibits neither convexity nor concavity.
\end{document}